\documentclass[11pt,a4]{article}
\usepackage{mathtools}
\usepackage{xspace}
\usepackage{tikz}
\usepackage{amsthm}
\usepackage{amssymb}
\usepackage{authblk}
\usepackage{hyperref}
\usepackage{thmtools}
\usepackage{fullpage}
\usepackage[utf8]{inputenc}
\usetikzlibrary{patterns}

\theoremstyle{plain}
\newtheorem{definition}{Definition}[section]
\newtheorem{lemma}[definition]{Lemma}
\newtheorem{theorem}[definition]{Theorem}

\newtheorem{observation}[definition]{Observation}
\newtheorem{corollary}[definition]{Corollary}
\newtheorem{algo}[definition]{Algorithm}

\newcommand{\eps}{\varepsilon}
\newcommand{\bigo}{\mathcal{O}}
\newcommand{\Oh}{\mathcal{O}}

\newcommand{\norm}[1]{\ensuremath{\lVert#1\rVert}}
\newcommand{\ceil}[1]{\ensuremath{\lceil#1\rceil}}
\newcommand{\expect}[1]{\ensuremath{\mathbb{E}\left[#1\right]}}
\newcommand{\var}[1]{\ensuremath{\mathrm{Var}\left[#1\right]}}
\newcommand{\eqtop}[1]{\stackrel{\mathclap{{\mbox{\ensuremath{#1}}}}}{=}}
\newcommand{\eqeps}{\eqtop{\varepsilon}}
\newcommand{\ham}{\mathsf{Ham}}

\newcommand{\eSketch}{\mathsf{eSketch}}
\newcommand{\hSketch}{\mathsf{hSketch}}
\newcommand{\mSketch}{\mathsf{mSketch}}
\newcommand{\pSketch}{\mathsf{Sketch_p}}

\newcommand{\hSub}{\mathsf{hSub}}
\newcommand{\mSub}{\mathsf{mSub}}
\newcommand{\rSub}{\mathsf{Sub}}
\newcommand{\MI}{\mathsf{MI}}

\title{\texorpdfstring{$L_p$}{Lp} Pattern Matching in a Stream}

\author[1]{Tatiana Starikovskaya\thanks{This work was partially funded by the grant ANR-19-CE48-0016 from the French National Research Agency (ANR).}}

\author[2]{Michal Svagerka}

\author[3]{Przemys\l{}aw Uzna\'nski\thanks{Supported by Polish National Science Centre grant 2019/33/B/ST6/00298.}}

\affil[1]{DIENS, \'{E}cole normale sup\'{e}rieure, PSL Research University, France}
\affil[2]{ETH Z\"urich, Switzerland}
\affil[3]{Institute of Computer Science, University of Wroc\l{}aw, Poland}

\date{}

\begin{document}
\maketitle
\begin{abstract}
We consider the problem of computing distance between a pattern of length $n$ and all $n$-length subwords of a text in the streaming model.

In the streaming setting, only the Hamming distance ($L_0$) has been studied. It is known that computing the exact Hamming distance between a pattern and a streaming text requires $\Omega(n)$ space (folklore). Therefore, to develop sublinear-space solutions, one must relax their requirements. One possibility to do so is to compute only the distances bounded by a threshold $k$, see~[SODA'19, Clifford, Kociumaka, Porat] and references therein. The motivation for this variant of this problem is that we are interested in subwords of the text that are similar to the pattern, i.e. in subwords such that the distance between them and the pattern is relatively small. 

On the other hand, the main application of the streaming setting is processing large-scale data, such as biological data. Recent advances in hardware technology allow generating such data at a very high speed, but unfortunately, the produced data may contain about 10\% of noise~[Biol. Direct.'07, Klebanov and Yakovlev]. To analyse such data, it is not sufficient to consider small distances only. A possible workaround for this issue is the $(1\pm\eps)$-approximation. This line of research was initiated in [ICALP'16, Clifford and Starikovskaya] who gave a $(1\pm\eps)$-approximation algorithm with space~$\widetilde\Oh(\eps^{-5}\sqrt{n})$.

In this work, we show a suite of new streaming algorithms for computing the Hamming, $L_1$, $L_2$ and general $L_p$ ($0 < p < 2$) distances between the pattern and the text.  Our results significantly extend over the previous result in this setting. In particular, for the Hamming distance and for the $L_p$ distance when $0 < p \le 1$ we show a streaming algorithm that uses $\widetilde\Oh(\eps^{-2}\sqrt{n})$ space for polynomial-size alphabets. 
\end{abstract}

\section{Introduction}
In the problem of pattern matching, we are given a pattern $P$ of length $n$ and a text $T$ and must find all occurrences of $P$ in $T$. A particularly relevant variant of this fundamental question is approximate pattern matching, where the goal is to find all subwords of the text that are similar to the pattern. This can be restated in the following way: given a pattern $P$, a text $T$, and a distance function, compute the distance between $P$ and every $n$-length subword of $T$. A very natural similarity measure for words is the Hamming distance. Furthermore, if both $P$ and $T$ are over an integer alphabet $\Sigma$, one can consider the Manhattan distance or the Euclidean distance.

\begin{definition}[Hamming, Manhattan and Euclidean distances]
For a vector $U = u_1 u_2 \ldots u_n$, its Hamming norm is defined as $\norm{U}_H = |\{i : u_i \not= 0\}|$, Manhattan norm is defined as $\norm{U}_1 = \sum_i |u_i|$ and Euclidean norm is defined as $\norm{U}_2 = \left( \sum_i u_i^2 \right)^{1/2}$. For two words $V = v_1 v_2 \ldots v_n$ and $W = w_1 w_2 \ldots w_n$, their Hamming distance is defined as $\norm{V-W}_H$, their Manhattan distance as $\norm{V-W}_1$, and their Euclidean distance as $\norm{V-W}_2$.
\end{definition}

Those distance functions naturally generalize to the so called $L_p$ distances, where $p > 0$ is the exponent. 

\begin{definition}[$p$'th moment, $p$'th norm]
For a vector $U = u_1 u_2 \ldots u_n$ and $p \ge 0$, its \emph{$p$'th moment} $F_p$ is defined as $F_p(U) = \sum_i |u_i|^p$, and for $p>0$ its \emph{$L_p$ norm} is defined as $\norm{U}_p = F_p(U)^{1/p} = \left(\sum_i |u_i|^p \right)^{1/p}$. For two words $V = v_1 v_2 \ldots v_n$ and $W = w_1 w_2 \ldots w_n$ considered as vectors, the $p$'th moment of their difference is $F_p(V-W)$ and their \emph{$L_p$ distance} is defined as $\norm{V-W}_p = F_p(V-W)^{1/p} = \left(\sum_i |v_i-w_i|^p\right)^{1/p}$.
\end{definition}
In other words, the Manhattan distance is the $L_1$ distance, the Euclidean distance is the $L_2$ distance, and the Hamming distance can be considered as the $L_0$ distance. 

Below we assume that the length of the text is $2n$, as any algorithm on a text of larger length can be reduced to repeated application of an algorithm that runs on texts of length~$2n$. This is done by splitting the text into blocks of length $2n$ which overlap by $n$ characters. 

\subparagraph*{Offline setting.}
For the Hamming distance, the problem has been extensively studied in the offline setting, where we assume random access to the input. The first algorithm, for a constant-size alphabet, was shown by Fischer and Paterson~\cite{FP:1974}. The algorithm uses $\Oh(n \log n)$ time and in substance computes the Boolean convolution of two vectors a constant number of times. This was later extended to polynomial-size alphabets in \cite{Abrahamson87,K:1987}. With a somewhat similar approach,
the same complexity can be achieved for the $L_1$ distance in \cite{CCI:2005}. Later, in \cite{GLU:2018,DBLP:journals/ipl/LipskyP08a} the authors proved that these problems must have equal (up to polylogarithmic factors) complexities by showing reductions from the Hamming to the $L_1$ distance and back. 

To improve the complexity for large alphabets, the natural next step was to study approximation algorithms. Until very recently, the fastest $(1\pm\eps)$-approximation algorithm for computing the Hamming distances was by Karloff~\cite{Karloff93}. The algorithm combines random projections from an arbitrary alphabet to the binary one and Boolean convolution to solve the problem in $\Oh(\eps^{-2} n \log^3 n)$ time. In a breakthrough paper Kopelowitz and Porat \cite{KP:15} gave a new approximation algorithm  improving the time complexity to $\Oh (\eps^{-1} n\log^3{n}\log{\eps^{-1}} )$, which was later significantly simplified~\cite{KopelowitzP18}. Using a similar technique, Gawrychowski and Uzna\'nski~\cite{GU18} showed an approximation algorithm for computing the $L_1$ distance in $\Oh(\eps^{-1} n \log^4 n)$ (randomized) time, later made deterministic in time $\Oh(\eps^{-1} n \log^2 n)$ in \cite{cpm19a}. Using similar techniques, the authors of \cite{cpm19a} gave $\widetilde{\bigo}(\eps^{-1} n)$-time $(1+\eps)$-approximation algorithm for $L_p$ distances for any constant positive $p$.\footnote{Across the paper we use $\widetilde{\bigo}$ to indicate that we are suppressing poly-log(n) factors.}

\subparagraph*{Streaming setting.} In the streaming setting, we assume that the pattern and the text arrive as streams, one character at a time (the pattern arrives before the text). The main objective is to design algorithms that use as little space as possible, and we must account for all the space used by the algorithm, including the space required to store the input, in full or in part. It is also often the case that the text arrives at a very high speed and we must be able to process it faster than it arrives to fulfil the space guarantees, preferably, in real time. To this aim, the time complexity of streaming algorithms is defined as the worst-case amount of time spent on processing one character of the text, i.e. \emph{per arrival}.

In the streaming setting, only the Hamming distance ($L_0$) has been studied. It is known that computing the Hamming distance between a pattern and a streaming text exactly requires $\Omega(n)$ space, even for the binary alphabet and with a small probability error allowed, which can be shown by a straightforward reduction to communication complexity (folklore). 

Therefore, to develop sublinear-space solutions, one must relax their requirements. One possibility to do so is to compute only the distances bounded by a threshold $k$. This variant of the problem is often reffered to as \emph{$k$-mismatch problem}. The $k$-mismatch problem has been extensively studied in the literature~\cite{k-mismatch,DBLP:conf/soda/CliffordKP19,DBLP:conf/icalp/GolanKP18,Porat:09}, with this line of work reaching $\widetilde\bigo(k)$ memory complexity and $\widetilde\bigo(\sqrt{k})$ time per input character. The motivation for this variant of this problem is that we are interested in subwords of the text that are similar to the pattern, in other words, the distance between the pattern and the text should be relatively small. On the other hand, the main application of the streaming setting is processing large-scale data, such as biological data. To decrease the cost of generating such data, recently new hardware approaches have been developed. They have become widely used due to cost efficiency, but unfortunately, the produced data may contain about 10\% of noise~\cite{noise-biology}. To analyse such data, it is not sufficient to consider small distances only, and a possible workaround for this issue is $(1\pm\eps)$-approximation. This line of research was initiated by Clifford and Starikovskaya~\cite{HDstream} who gave a $(1\pm\eps)$-approximation algorithm with space $\widetilde\Oh(\eps^{-5}\sqrt{n})$ that uses $\widetilde\Oh(\eps^{-4})$ time per arriving character of the text.

Independently and in parallel with this work, authors of \cite{approxk} showed a $(1\pm\eps)$-approximation streaming algorithm for the $k$-mismatch problem that uses $\widetilde\bigo(\eps^{-2}\sqrt{k})$ space. For a special case of $k = n$, they show how to reduce the space further to $\widetilde\bigo(\eps^{-1.5}\sqrt{n})$. Compared to our solution, their algorithm has worse time complexity of $\widetilde\bigo(\eps^{-3})$ per arrival, and more importantly, it is not obvious whether it can be generalised to other $L_p$ norms as it uses a very different set of techniques.

\subparagraph*{Sliding window.} The problem of computing distance between $P$ and every $n$-length subword of $T$ in the streaming setting resembles the problem of maintaining the $L_p$ norm of a $n$-length suffix of a streaming text, also referred to as \emph{sliding window}. In fact, the latter is a simplification of the former, with setting $P = [0,0,\ldots,0]$. There is an extensive line of work on maintaining the $L_p$ norm of a sliding window, refer to~\cite{DBLP:journals/tcs/BravermanGO14,DBLP:conf/focs/BravermanO07,DBLP:journals/siamcomp/BravermanO10,DBLP:conf/approx/BravermanOR15,DBLP:journals/jcss/BravermanOZ12,DBLP:journals/siamcomp/DatarGIM02} and references therein. The main message is that the norm of a sliding window can be maintained efficiently, e.g. for $1 \le p \le 2$ the $L_p$ norms can be maintained $(1\pm\eps)$-approximately in space $\widetilde\bigo(\eps^{-1})$. However, those results do not translate to our case: in the sliding window, one can easily isolate ``heavy hitters'', that is updates with a significant contribution to the output. In our case, the contribution of an update depends on its relative position to the pattern, and one can easily construct instances where a contribution of a position in the text changes drastically relative to its alignment with the pattern, which necessitates a significantly different approach.

\subsection{Our results}
In this work, we show a suite of new streaming algorithms for computing the  Hamming, $L_1$, $L_2$ and general $L_p$ ($0 < p \le 2$) distances between the pattern and the text. Our results significantly improve and extend the results of~\cite{HDstream}.

\begin{theorem}
\label{th:UB}
Given a pattern $P$ of length $n$ and a text $T$ over an alphabet $\Sigma = [1,2,\ldots,\sigma]$, where $\sigma = n^{\Oh(1)}$, there is a streaming algorithm that computes a $(1\pm\eps)$-approximation of the $L_p$ distance between $P$ and every $n$-length subword of $T$ correctly w.h.p.
\begin{enumerate}
\item \label{it:UB_ham} in $\widetilde\Oh(\eps^{-2} \sqrt n + \log \sigma)$ space, and $\widetilde\Oh(\eps^{-2})$ time per arrival when $p = 0$ (Hamming distance);
\item \label{it:UB_l1} in $\widetilde{\Oh}(\eps^{-2} \sqrt{n} + \log^2 \sigma)$ space and $\widetilde{\Oh}(\sqrt{n} \log \sigma)$ time per arrival when $p = 1$ (Manhattan distance);
\item in $\widetilde{\Oh}(\eps^{-2} \sqrt{n} + \log^2 \sigma)$ space and $\widetilde{\Oh}(\eps^{-2} \sqrt{n})$ time per arrival when $0 < p < 1/2$;
\item in $\widetilde{\Oh}(\eps^{-2} \sqrt{n} + \log^2 \sigma)$ space and $\widetilde{\Oh}(\eps^{-3} \sqrt{n})$ time per arrival when $p = 1/2$;
\item in $\widetilde{\Oh}(\eps^{-2} \sqrt{n} + \log^2 \sigma)$ space and $\widetilde{\Oh}(  \sigma^{\frac{2p-1}{1-p}} \sqrt{n} / \eps^{2+3 \cdot \frac{2p-1}{1-p}} )$ time per arrival when $1/2 < p < 1$;
\item \label{it:UB_lp} in $\widetilde{\Oh}(\eps^{-2-p/2} \sqrt{n} \log^2 \sigma)$ space and $\Oh(\eps^{-p/2} \sqrt{n} + \eps^{-2} \log \sigma)$ time per arrival for $1 < p \le 2$.
\end{enumerate}
\end{theorem}

We also improve and extend the space lower bound of~\cite{HDstream}, who showed that any streaming algorithm that computes a $(1\pm\eps)$-approximation of the Hamming distance between a pattern and a streaming text must use $\Omega(\eps^{-2} \log^2 n)$ bits for all $\eps$ such that $1/\eps < n^{1/2-\gamma}$ for some constant $\gamma$ (condition inherited from~\cite{JW2013}). We show the following result: 

\begin{lemma}
\label{lm:LB_ham_L1}
Let $2 \le 1/\eps < n$ and $0 \le p \le 2$. Any $(1\pm\eps)$-approximation algorithm that computes the $L_p$ distance between a pattern and a streaming text for each alignment, must use $\Omega(\min(1/\eps^2, n))$ bits of space.
\end{lemma}
\begin{proof}
Let us first show the lower bound for $p = 0$, i.e., for Hamming distance. We show the lower bound by reduction to a  two-party communication complexity problem called GAP-Hamming-distance. In this problem, the two parties, Alice and Bob are given two binary words of length $n$ and a parameter $g = \eps n$, $1 \le g \le n/2$. Alice sends Bob a message, and Bob's task is to output $1$ if the Hamming distance between his and Alice's word is larger than $n/2+g$, and zero if it is at most $n/2-g$. Otherwise, he can output ``don't know''.
By Proposition 4.4~\cite{GAP-Hamming}, the communication complexity of this problem is $\Omega(\min\{1/\eps^2, n\})$. 

We can now show a space lower bound for any $(1\pm\eps)$-approximate algorithm for computing the Hamming distance between the pattern and the text by a standard reduction. Suppose that $2 \le 1/\eps \le n$ there is an algorithm that uses $o(\min\{1/\eps^2, n\})$ bits of space. Let $P$ be Alice's word, $T$ Bob's word. After reading~$P$, the algorithm stores all the information about it in $o(\min\{1/\eps^2, n\})$ bits of space. We construct the communication protocol as follows: Alice sends the information about $P$ to Bob. Using it, Bob can continue running the algorithm and compute the approximation of the Hamming distance between $P$ and~$T$. We have thus developed a communication protocol with complexity $o(\min\{1/\eps^2, n\})$, a contradiction.

We can now show the lower bound for $0 < p \le 2$. We immediately obtain a space lower bound for any $(1\pm\eps)$-approximate algorithm for computing the $p$'th moment between the pattern and the text at every alignment. Indeed, on binary words the $p$'th moment is equal to the Hamming distance for all $0 < p \le 2$. The lower bound for the $L_p$ distance follows by Observation~\ref{obs:norm_moment}.
\end{proof}

\subsection{Techniques}\label{sec:structure}
At a very high level, the structure of all algorithms presented in this paper is similar to that of~\cite{HDstream} (in fact, such approach in similar context was also used independently in \cite{DBLP:conf/approx/CrouchM11}). We process the text by blocks of length $b \approx \sqrt{n}$.  To compute an approximation of the distance / the $p$'th moment at a particular alignment, we divide the pattern into two parts: a prefix of length $\le b$ aligned with a suffix of some block of the text, and the remaining suffix (see Fig.~\ref{fig:alg_structure}). We compute an approximation of the distance / the $p$'th moment for both of the parts and sum them up to obtain the final answer. Our main contribution is a set of new tools that allows computing the approximations efficiently.

To be able to compute the approximation of the distance / the $p$'th moment between the prefix and the corresponding block of the text, we compute, while reading each block of the text, its compact lossy description that we refer to as \emph{prefix encoding}. The prefix encoding captures the relation between the read block and the prefix of the pattern of length $b$. To compute the distance / the $p$'th moment between the suffix and the text, we will use \emph{suffix sketches}. For each position $i$ of the text, the suffix sketch describes the subword $T[b\cdot k + 1, i]$ of the text where $k$ is the smallest integer such that $i - b \cdot k \le n$  (see Fig.~\ref{fig:alg_structure}). 

\begin{figure}
\begin{center}
\begin{tikzpicture}[scale=0.45]	
	\useasboundingbox (10,3.5) rectangle (40,9);	

	\draw[thick] (38,5.5) rectangle (18,4.5);
	\draw[thick,pattern=north east lines, pattern color=gray] (20,4.5) rectangle (38,5.5);
	\draw[thick,pattern=north west lines, pattern color=red] (18,4.5) rectangle (20,5.5);
	\draw (20,4.5)--(20,5.5);
		
	\node[below] at (19,4.5) {\small $P[1,j]$};
	\node[below] at (29,4.5) {\small $P[j+1,n]$};
	\node[above] at (19.7,8.3) {\small $bk$};
	\node[above] at (37.7,8.3) {\small $i$};
	
	\begin{scope}[yshift = 7.5cm]
		\foreach \x in {3,4} {
			\draw (\x*5,0)--(\x*5,1);
			
			\draw[thin] (\x*5-5,0)--(\x*5-5,-1);
			\draw[thin,<->] (\x*5-5,-0.8)--(\x*5,-0.8);
			\node[above] at (\x*5-2.5,-1) {\it \small prefix enc.};	
		}	
		\draw[thin] (20,0)--(20,-1);	
		\draw[thin] (38,0)--(38,-1);	
		\draw[thin,<-] (38,-0.8)--(20,-0.8);	
		\node[above] at (29,-1) {\it \small suffix sk.};
		
		\draw[thick,dashed] (38,0) rectangle (10,1);
		\draw[thick,pattern=north west lines, pattern color=gray] (20,0) rectangle (38,1);
		\draw[thick,pattern=north west lines, pattern color=red] (18,0) rectangle (20,1);		
	\end{scope}
\end{tikzpicture}
\end{center}
\caption{High level structure of the algorithms. To compute the distance between the prefix (red) of the pattern and the text, we use the prefix encoding, between the suffix (grey) and the text we use the suffix sketch.}
\label{fig:alg_structure}
\end{figure}
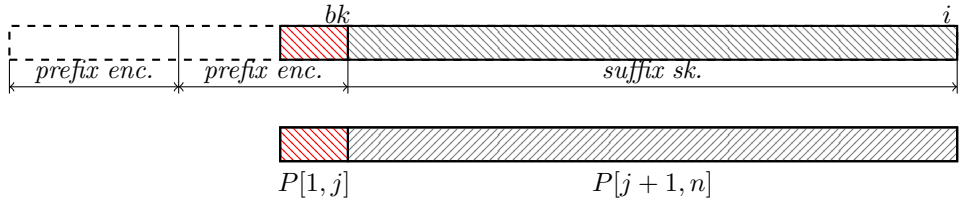

For the Hamming distance, we define the prefix encodings in Section~\ref{sec:prefix_l0} and the suffix sketches in Section~\ref{sec:suffix_l0}. Our Hamming prefix encoding introduces a novel use of a known technique called \emph{subsampling}. The prefix encodings are used to approximate the distance between any suffix of one word and the prefix of another word of the same length. In brief, the idea is to replace each character of the two words by the don't care character ``?'', a special character that matches any other character of the alphabet. We repeat the process a logarithmic number of times to create a logarithmic number of pairs of ``subsamples''. For each pair, we find the longest suffix of one subsample that matches the prefix of the second subsample 
up to at most $\Theta(1/\eps^2)$ mismatches. We then show that this information can be used to approximate the Hamming distance between any suffix-prefix pair. Similar techniques were used in~\cite{DBLP:conf/random/Bar-YossefJKST02,Durand2003,DBLP:journals/tcs/Ganguly07,DBLP:conf/spaa/GibbonsT01,DBLP:conf/pods/KaneNW10,DBLP:journals/siamcomp/PavanT07} for estimating the Hamming norm in streams. The crucial difference with our approach is that we must be able to compute the Hamming norm of any suffix-prefix pair of the two words, and we must be able to do it efficiently. As for the suffix sketches, for the binary alphabet we use the sketches introduced in~\cite{HDstream}. We then show a reduction from arbitrary alphabets to the binary alphabet, which improves the space consumption of Hamming suffix sketches by a factor of $1/\eps^2$. 

We can solve the problem of $L_1$ (Manhattan distance) pattern matching by replacing each character of the pattern and of the stream with its unary encoding and running the solution for the Hamming distance. However, this would introduce a multiplicative factor of $\sigma$ (the size of the alphabet) to the time complexity. We show efficient randomised reductions from the Manhattan to Hamming distance that allow simulating the solution for the Hamming distance without a significant overhead. In particular, to design the prefix encodings we use random shifting and rounding, while for the suffix sketches we use range-summable hash functions~\cite{DBLP:conf/soda/CalderbankGLMS05}. We show the Manhattan prefix encodings in Section~\ref{sec:prefix_l1} and the Manhattan suffix sketches in Section~\ref{sec:suffix_l1}. 

For generic $L_p$ distances, $0 < p \le 2$, we discuss the prefix encodings in Section~\ref{sec:prefix_lp_l2} and the suffix sketches in Section~\ref{sec:suffix_lp}. Our approach to $L_p$ prefix encodings is rather involved. In the case of $0<p< 1$, we construct a novel embedding from $L_p^p$ space into the Hamming space, which might be of independent interest. While the target dimension of the Hamming space is large, we construct the embedding in such a way that each value is mapped into a compressible sequence of form $c_1^{d_1} \ldots c_t^{d_t}$ for some small value of $t$, and where values of $d_1,\ldots,d_t$ are constant across all input values. Such compressed representation allows us to efficiently apply the subsampling framework and reduce the problem to the Hamming distance case. For $1<p\le 2$, we  identify a logarithmic number of anchor suffixes, and partition each of them  into $\eps^{-p}$ words of roughly even contribution to the distance. We then use the partition to decode prefix-suffix distance queries for arbitrary length queries. Such construction is a generalization and improvement of the approach presented in \cite{HDstream}.  For suffix sketches, we simply use the $p$-stable distributions~\cite{DBLP:journals/jacm/Indyk06}.

Finally, we combine the prefix encodings and the suffix sketches to prove Theorem~\ref{th:UB} in Section~\ref{sec:theorem}. To simplify the notation, we use $x \eqeps y$ to denote $(1-\eps) y  \le x \le (1+\eps) y$ from now on. We will also use the fact that for $p > 0$ we can speak of approximating the $p$'th moment of differences between the pattern and the $n$-length substrings of the text and the $L_p$ distances between the pattern and the $n$-length substrings of the text interchangeably, it changes the complexities up to a constant factor only:

\begin{observation}\label{obs:norm_moment}
For any constant $p > 0$ and $\varepsilon < 1/2$, there is a constant $C_p$ such that finding a $(1 \pm C_p \cdot p \varepsilon)$ approximation of the $p$'th moment of a vector suffices for $(1 \pm \varepsilon)$-approximating its $p$'th norm, and finding a  $(1 \pm C_p \cdot \varepsilon/p)$ approximation of its $p$'th norm suffices for $(1 \pm \varepsilon)$-approximating its $p$'th moment.
\end{observation}

\section{Prefix encodings}
In this section we present a solution to the following problem. Imagine we have a block of text $T'[1, b]=T[i+1, i+b]$ and a prefix of the pattern $P' = P[1, b]$. We want to find a compressed representation (encoding) of~$T'$ so that the following is possible: given any $1 \le d \le b$, the compressed representation of~$T'$, and $P'$ (explicitly), we can $1\pm\eps$ approximate $\norm{T''-P''}_p$, where $T'' = T'[b-d+1, b]$ is a suffix of $T'$ and $P'' = P'[1, d]$ is a prefix of $P'$. 

We start by presenting a solution to the Hamming distance case, which is a basis to our solution for all other $L_p$ norms for $0<p\le 2$.

\subsection{Hamming (\texorpdfstring{$L_0$}{L0}) distance}\label{sec:prefix_l0}
Recall that ``?'' is the don't care character, a special character that matches any other character of the alphabet.

\begin{definition}[Hamming subsampling]\label{def:subsampling} 
Consider a word $U$ of length $n$. Let $q = \ceil{3 \log n}$ and let $h(i) : [n] \rightarrow \{0,1\}^q$ be a function drawn at random from a pairwise independent family. For $r = 0, \ldots, q$, we define the $r$-th level Hamming subsample of $U$, $\hSub_r(U)$, as follows: 
\[\hSub_r(U) [i] = 
\begin{cases} 
U[i], &\text{if the } r\ \text{lowest bits of}\ h(i)\ \text{are all}\ 0;\\ 
?, &\text{otherwise.}
\end{cases}.
\] 
In particular, $\hSub_0(U) = U$.
\end{definition}

Fix an integer $k = \Theta(1/\eps^2)$ large enough. For two words $U,V$, consider the following estimation procedure:
\begin{algo}\label{alg:ham}\ 
\begin{enumerate}
\item Denote $X_r$ to be the Hamming distance between $\hSub_r(U)$ and $V$ and let $f = \min\{ i : X_i \le k\}$.\footnote{We emphasize that $\hSub_r(U)$ contains don't care characters, so the Hamming distance is defined as the number of pairs of characters of $\hSub_r(U)$ and $V$ that do not match.}
\item Output $Z_f = 2^f \cdot X_f$ as an estimate of $\norm{U-V}_H$. 
\end{enumerate}
\end{algo}

The following lemma is a rephrasing of a known result regarding subsampling in estimation of the Hamming norm (cf. \cite[Theorem 3]{DBLP:conf/random/Bar-YossefJKST02}, or \cite[Theorem 2]{DBLP:conf/spaa/GibbonsT01}). 

\begin{lemma}
\label{lem:hamming_subsampling}
For $Z_f$ as in Algorihtm~\ref{alg:ham} there is $Z_f \eqeps \norm{U-V}_H$ with probability at least $3/4$.
\end{lemma}
\begin{proof}
Denote $m = \norm{U-V}_H$. Consider a fixed value $r$. Let $I_1,I_2,\ldots,I_n$ be binary variables indicating existence of a mismatch between $\hSub_r(U)$ and~$V$ at positions $1,\ldots,n$, so that $X_r = \sum_j I_j$. We observe that $\expect{X_r} = m/2^r$ and therefore $\expect{Z_r} = m$, because each of the $m$ positions with mismatch between $U$ and $V$ generates a mismatch between $\hSub_r(U)$ and~$V$ with probability $1/2^r$. 

Furthermore, as the function $h$ in Definition~\ref{def:subsampling} is drawn from a pairwise independent family, there is $\var{X_r} = \sum_j \var{I_j}  \le \sum_j \expect{(I_j)^2}  = \sum_j \expect{I_j} = \expect{X_r} = m/2^r$. Let $c = \min \{i : \expect{X_i} \le k\} = \ceil{ \log_2 \left( \frac{m}{k} \right) }$. By Chebyshev's inequality, we have
\begin{equation}
\label{eq:z_large} \Pr[ |Z_r - m| \ge 4\sqrt{m 2^{c+1}}] = \Pr[ |X_r - m/2^r| \ge 2^{2+(c+1-r)/2} \sqrt{m/2^r}] \le 1/2^{4+(c+1-r)}
\end{equation}
We estimate $\Pr[ f > c+1] = \Pr[X_{c+1} > k]$. Assume w.l.o.g. that $k \ge 32$. Observe that $m/2^c \le k$, which implies, for $k \ge 32$, $m/2^{c+1} + 4\sqrt{m/2^{c+1}} \le k/2 + 4\sqrt{k/2} \le k$. By Equation~\ref{eq:z_large}, there is 
\[\Pr[X_{c+1} > k] \le \Pr[X_{c+1} \ge m/2^{c+1} + 4\sqrt{m/2^{c+1}}] \le 1/16.\]
It follows that $\Pr[ f > c+1] = \Pr[X_{c+1} > k] \le 1/16$. 
Hence, we obtain 
\begin{align*}
\Pr[ |Z_f - m| \ge 4 \sqrt{2/k} \cdot m] &\le \Pr[ |Z_f - m| \ge 4 \sqrt{m 2^{c+1}} ] \le\\
&\le \Pr[ f > c+1] + \sum_{r = 0}^{c+1} \Pr[|Z_f-m| \ge 4 \sqrt{m 2^{c+1}}  \mbox{ and } f = r] \le \\
&\le \Pr[ f > c+1] + \sum_{r = 0}^{c+1} \Pr[|Z_r-m| \ge 4 \sqrt{m 2^{c+1}}] \le \\
&\le 1/16 + \sum_{r=1}^{c+1} 1/2^{4+(c+1-r)} < 1/4.
\end{align*}
It follows that we can choose $k = \Theta(1/\eps^2)$ large enough so that $Z_f \eqeps \norm{U-V}_H$ with probability $\ge 3/4$. 
\end{proof}
Since the subsampling is performed independently for each position, one can use subsampling to approximate the Hamming distance between any suffix of $B$ and any prefix of $P$ of equal lengths in a similar fashion. 

We are now ready to define the Hamming prefix encoding of a block. For brevity, let $B_r^{j} =\hSub_r(B)[b-j+1,b]$ and $P_r^j = P[1,j]$ (the same for all $r$). Furthermore, given two words $U,V$ of equal length, define the \emph{mismatch information} $\MI(U, V) = \{(i, U[i], V[i]) : U[i] \mbox{ does not match } V[i]\}$. 

\begin{definition}\label{def:prefix_encoding_l0}
Consider a $b$-length block $B$ of the text $T$. For each $0 \le r \le \ceil{3 \log n}$, let $j^\ast (r)$ be the maximal integer such that the Hamming distance between $B_r^{j^\ast (r)}$ and $P_r^{j^\ast (r)}$ is at most $k = \Theta(\varepsilon^{-2})$. We define the Hamming prefix encoding of $B$ to be a tuple of pairs $j^\ast (r), \MI(B_r^{j^\ast (r)}, P_r^{j^\ast (r)})$. 
\end{definition}

Note that the prefix encoding of $B$ uses $\bigo (k \log n) = \bigo (\eps^{-2} \log n)$ space. We can compute it efficiently:

\begin{lemma}\label{lm:hprefsketchalg}
Assume constant-time random access to $P[1,b]$. Given a $b$-length block $B$ of the text $T$, its Hamming prefix encoding can be computed in $\widetilde\Oh(k b)= \widetilde\bigo(b \varepsilon^{-2})$ time.
\end{lemma}
\begin{proof}
To compute the encoding, we use the algorithm of~\cite{DBLP:conf/soda/CliffordEPR09}. Formally, for each $r$ we create a word $T'$ by appending $b$ don't care characters to the subsample $\hSub_r$. The algorithm of ~\cite{DBLP:conf/soda/CliffordEPR09} can be used to find all $b$-length subwords of $T'$ that match $P[1,b]$ with up to $k$ mismatches, moreover for each of these subwords the algorithm outputs the mismatch information. We take the leftmost subword only, which corresponds to $j^\ast (r)$ because of the don't care characters. In total, our algorithm uses $\widetilde\Oh(kb) = \widetilde\Oh(\eps^{-2} b)$ time. 
\end{proof}

We now show how to compute the Hamming distance between any $j$-length suffix of $B$ and any $j$-length prefix of~$P$ given $P[1,b]$ and the Hamming prefix encoding of a block $B$.

\begin{lemma}\label{lm:hpref_decoding}
Given the prefix encoding of a $b$-length block $B$ of the text $T$, there is an algorithm that computes, for any $j = 1, \ldots, b$, a $(1+\eps)$-approximation of the Hamming distance between the $j$-length suffix of $B$ and the $j$-length prefix of $P$ in $\widetilde\Oh(k b) = \widetilde\bigo(b \varepsilon^{-2})$ time. 
\end{lemma}
\begin{proof}
Denote $X_r$ to be the Hamming distance between $P_r^j$ and $B_r^j$. We compute the smallest $f$ such that $X_f \le k$ in the following way. For each $r$, we use $\MI(B_r^{j^\ast (r)}, P_r^{j^\ast (r)})$ to restore $B_r^{j^\ast (r)}$. We then append $P_r^{j^\ast (r)}$ with $b$ don't care characters and run the algorithm of~\cite{DBLP:conf/soda/CliffordEPR09} for the resulting text and the pattern. This allows to compute $X_r$ for all $j \le j^\ast (r)$, and if $j > j^\ast (r)$, then $X_f > k$ by definition. In total, the algorithm takes $\widetilde\bigo(kb) = \widetilde\Oh(\eps^{-2} b)$ time.
\end{proof}

\subsection{Manhattan (\texorpdfstring{$L_1$}{L1}) distance}\label{sec:prefix_l1}
Recall a word morphism $\nu : \Sigma \to \{0,1\}^{\sigma}$, $\nu(a) = 1^{a}0^{\sigma-a}$. Our goal in this section is to simulate implicitly procedures from Lemma~\ref{lm:hprefsketchalg} and Lemma~\ref{lm:hpref_decoding} on words $\nu(B)$ and $\nu(T)$ without introducing any significant overhead.

\begin{definition}[Manhattan scaling]\label{def:manhattanscaling} 
Consider a word $U$ of length $n$. Let $q = \ceil{3 \log n \sigma}$ and let $h : [n] \rightarrow 2^q$ be a function drawn at random from a $4$-wise independent family. For $r = 0, \ldots, q$, we define the $r$-th level Manhattan subsample of $U$, $\mSub_r(U)$, as a word of length $n$ such that $\mSub_r(U)[i] = \left\lfloor \frac{U[i] + (h(i) \bmod 2^r)}{2^r} \right\rfloor$. 
In particular, $\mSub_0(U) = U$.
\end{definition}

Fix an integer $k = \Theta(1/\eps^2)$ large enough. For words $U,V$, consider $\mSub_r(U), \mSub_r(V)$ for all $r = 0, \ldots, q$, and the following estimation procedure:
\begin{algo}\label{alg:l1pref}\ 
\begin{enumerate}
\item Denote $X_r = \norm{\mSub_r(U)-\mSub_r(V)}_1$ and let $f = \min\{ i : X_i \le k\}$.
\item Output $Z_f = 2^f \cdot X_f$ as an estimate of $\norm{U-V}_1$. 
\end{enumerate}
\end{algo}

\begin{lemma}
\label{lem:manhattan_scaling}
For $Z_f$ as in Algorihtm~\ref{alg:l1pref} there is $Z_f \eqeps \norm{U-V}_1$ with probability $\ge 3/4$.
\end{lemma}

\begin{proof}
Take some position $i$ and denote for short $a = \mSub_r(U)[i]$ and $b = \mSub_r(V)[i]$ and $c = \frac{U[i] - V[i]}{2^r}$. There is  $  |a - b|  \in \left\{ \big\lfloor |c| \big\rfloor,  \big\lceil |c| \big\rceil \right\}$ and 
$\expect{|a-b|} =|c|.$ Since $|a-b|-\big\lfloor |c| \big\rfloor$ is a $0/1$ variable, there is $\var{|a-b|} = \var{\left( |a-b|-\big\lfloor |c| \big\rfloor \right)} \le \expect{\left( |a-b|-\big\lfloor |c| \big\rfloor \right)}  \le \expect{|a-b|}$. Summing for all values of $i$, we reach that
\[\var{X_r} = \var{ \norm{\mSub_r(U) - \mSub_r(V)}_1 } \le \expect{ \norm{\mSub_r(U) - \mSub_r(V)}_1 } =  \expect{X_r}.\]
Since we have reached an identical variance bound, the proof follows step-by-step the proof of Lemma~\ref{lem:hamming_subsampling}.
\end{proof}

To approximate the Manhattan distance between any suffix of $B$ and any prefix of $P$ of equal lengths, we define the encoding similar to the Hamming distance case. Specifically, we still use the \emph{mismatch information}, building on the fact that for any two words $\norm{U-V}_H \le \norm{U-V}_1$ and from the mismatch information the exact value of $\norm{U-V}_1$ can be found. We define $B_r^j = \mSub_r(B)[b-j+1,b]$ as before, but change the definition of $P_r^j$ slightly. Intuitively, we define $P_r^j$ to be the $j$-length prefix of $P$ subsampled in a synchronized way with $B_r^j$. Formally, $P_r^j [i] = \left\lfloor \frac{P[i] + (h(b-j+i) \bmod 2^r)}{2^r} \right\rfloor$.
 
\begin{definition}\label{def:prefix_encoding_l1}
Consider a $b$-length block $B$ of the text $T$. For each $0 \le r \le \ceil{3 \log n \sigma}$, let $j^\ast (r)$ be the maximal integer such that the Manhattan distance between $B_r^{j^\ast (r)}$ and $P_r^{j^\ast (r)}$ is at most $k = \Theta(\varepsilon^{-2})$. We define the Manhattan prefix encoding of $B$ to be a tuple of pairs $j^\ast (r), \MI(B_r^{j^\ast (r)}, P_r^{j^\ast (r)})$. 
\end{definition}

Note that the prefix encoding of $B$ uses $\bigo (k \log n\sigma) = \bigo (\eps^{-2} \log n)$ space.

\begin{lemma}\label{lm:mprefsketchalg}
Assume constant-time random access to $P[1,b]$. Given a $b$-length block $B$ of the text $T$, its Manhattan prefix encoding can be computed in $\widetilde\Oh(b^{2})$ time and $\widetilde\Oh(b)$ space.
\end{lemma}
\begin{proof}
Let $q = \ceil{3 \log n \sigma}$. For each $r = 0, \ldots, q$ and $j = 1,\ldots, b$ we compare $B_r^j$ and $P_r^j$ character by character in $\Oh(b)$ time to find $j^\ast(r)$ and the corresponding mismatch information. The claim follows. 
\end{proof}

\begin{lemma}\label{lm:mpref_decoding}
Given the prefix encoding of a $b$-length block $B$ of the text $T$, there is an algorithm that computes, for all $j = 1, \ldots, b$, a $(1\pm\eps)$-approximation of the Manhattan distance between the $j$-length suffix of $B$ and the $j$-length prefix of $P$ in $\widetilde\Oh(b^{2})$ time.
\end{lemma}
\begin{proof}
Denote $X_r = \norm{P_r^j - B_r^j}_H$. We compute the smallest $f$ such that $X_f \le k$ in the following way. For each $r$, we use $\MI(B_r^{j^\ast (r)}, P_r^{j^\ast (r)})$ to restore $B_r^{j^\ast (r)}$. If $j > j^\ast (r)$, the Manhattan distance between $P_r^j$ and $B_r^j$ is at least $k$. Otherwise, we compare $P_r^j$ and $B_r^j$ character by character to compute the Manhattan distance in $\Oh(b)$ time. The claim follows.
\end{proof}

\subsection{Generic (\texorpdfstring{$L_p$}{Lp}) distance for \texorpdfstring{$0 < p < 1$}{0 < p < 1}}

Our goal is to construct a morphism (parametrised by $p$) acting as a randomized embedding of $(L_p)^p$ into the Hamming distance. The intuition behind our approach is as follows. Let $r_0,r_1,\ldots \in [0,1]$ be a sequence of real numbers picked independently and u.a.r. Define a sequence of values 
\[d_i = \begin{cases}\varepsilon^{-1} \cdot (1+\varepsilon)^{pi}&\text{ when } i>0\\ \varepsilon^{-1} \cdot \frac{(1+\varepsilon)^p}{(1+\varepsilon)^p-1}&\text{ when } i=0\end{cases}\]
and for a character $c \in \Sigma$  consider  sequence of characters $c_0, c_1, \ldots$ where $c_i = \lfloor \frac{c}{(1+\varepsilon)^i} + r_i\rfloor$ (similarly, a character $c'$ defines a sequence $c'_0, c'_1, \ldots$).
Now consider two characters $c,c' \in \Sigma$ such that $|c-c'| = (1+\varepsilon)^\ell$ for some integer $\ell$ and a random variable $x = \sum_{i=0}^{\infty} d_i \cdot \norm{c_i - c'_i}_H$. There is 
\begin{align}
\expect{x} &= \sum_{i=0}^{\infty} d_i \cdot \Pr[c_i \not= c'_i] = \sum_{i=0}^{\ell} d_i \cdot 1 + \sum_{i=\ell+1}^{\infty} d_i \cdot \frac{|c-c'|}{(1+\varepsilon)^i} = \nonumber\\ 
 &= \varepsilon^{-1}\frac{(1+\varepsilon)^p}{(1+\varepsilon)^p-1}+\varepsilon^{-1}\sum_{i=1}^{\ell} \big((1+\varepsilon)^p\big)^i + \varepsilon^{-1} |c-c'| \sum_{i=\ell+1}^{\infty} \left((1+\varepsilon)^{p-1}\right)^i = \nonumber\\ 
 &= \varepsilon^{-1}\frac{((1+\varepsilon)^{\ell+1})^{p}}{(1+\varepsilon)^p-1} + |c-c'| \varepsilon^{-1} \frac{((1+\varepsilon)^{\ell+1})^{p-1}}{1 - (1+\varepsilon)^{p-1}} = \nonumber\\ 
 &=|c-c'|^{p} \left( \frac{(1+\varepsilon)^p }{(1+\varepsilon)^p-1} +  \frac{ 1}{(1+\varepsilon)^{1-p} - 1} \right) \varepsilon^{-1} \approx \varepsilon^{-2} |c-c'|^p  \frac{1}{p(1-p)}.\label{eq:expected_lp}
 \end{align}

We thus see that an idealized morphism of the form $\varphi: c \to c_0^{d_0}c_1^{d_1} \ldots$ would have the property that $\norm{U-V}_p^p  \sim \norm{\varphi(U)-\varphi(V)}_H $ on words of length $n$. But there are the following issues: (i) characters are mapped into infinite length words, (ii) number of repetitions of characters ($d_i$) is fractional, (iii) we cannot guarantee that character distance is always of form $(1+\varepsilon)^i$ and (iv) the distance is preserved only in expectation. We show how to overcome these issues to achieve the following result:

\begin{theorem}
\label{th:wordmorphism_p12}
Given $0 < p < 1$ and $\eps >0$ there is a word morphism $\varphi: c \in \Sigma \to c_0^{d_0}c_2^{d_2} \ldots c_{t-1}^{d_{t-1}}$ such that:
\begin{enumerate}
\item $t = \widetilde\bigo(\varepsilon^{-2})$ when $0 < p < 1/2$, $t = \widetilde\bigo(\varepsilon^{-3})$ when $p=1/2$ and $t = \widetilde\bigo( \sigma^{\frac{2p-1}{1-p}} / \eps^{2+3 \cdot \frac{2p-1}{1-p}})$ when $1/2 < p < 1$.
\item values of $t$ and $d_0,\ldots,d_{t-1}$ do not depend on $c$,
\item there exists a constant $\alpha = \alpha(p, \eps)$ such that for any two words $U,V$ of length at most $n$, we have $\norm{U-V}_p^p \eqeps \alpha \cdot \norm{\varphi(U)-\varphi(V)}_H$ with probability at least $9/10$,
\item it is enough for the randomness to be realized by a hash function $r: [t] \to [D]$ from a $4$-independent hash function family for some $D = \textrm{poly}(n \sigma \varepsilon^{-1})$, which can be generated from a $\widetilde\bigo(\log \sigma)$ bits size seed.
\end{enumerate}
\end{theorem}
\begin{proof}
We will consider three cases: $0 < p < 1/2$, $p = 1/2$, and $1/2 < p < 1$.

\textbf{Case $0 < p < 1/2$.} Our plan is to build upon the scheme highlighted earlier in this section. Specifically, we preserve the values of $c_i$.

Consider a pair of characters  $c, c'$. First, note that $\expect{x}$ is an increasing function of $|c-c'|$. From this and Equation~\ref{eq:expected_lp} we obtain that $\expect{x} \eqeps |c-c'|^{p} \left( \frac{(1+\varepsilon)^p }{(1+\varepsilon)^p-1} +  \frac{ 1}{(1+\varepsilon)^{1-p} - 1} \right) \varepsilon^{-1}$ for all values of $|c-c'|$.

Second, fix $q = \ceil{\frac{1}{1-p}\log_{1+\varepsilon}(\sigma \varepsilon^{-3})}$ and observe that  truncating the sum after the $(q-1)$-th term introduces an additional factor $1 \pm \Theta(\eps)$ to the approximation, since for $c \not= c'$ we have
\[\sum_{i \ge q} d_i \cdot \frac{|c-c'|}{(1+\varepsilon)^i} = \varepsilon^{-1} |c-c'| \frac{((1+\eps)^q)^{p-1}}{1-(1+\eps)^{p-1}} \le \frac{\eps^{-1} \sigma}{(1-(1+\eps)^{p-1}) \sigma \eps^{-3}} = \Theta(\eps).\]
We also round $d_i$ down to the nearest integer, which introduces an additional $1\pm \Theta(\varepsilon)$ relative error, since $\forall_i d_i \ge \varepsilon^{-1}$. Finally, we set $\varphi(c) = c_0^{d_0} \ldots c_{q-1}^{d_{q-1}}$. We then have $\expect{\norm{\varphi(c)-\varphi(c')}_H} = \Theta(\varepsilon^{-2}  |c-c'|^p \frac{1}{p(1-p)}).$

To guarantee that the equality holds with probability at least $9/10$ and not just in expectation, we repeat the scheme several times, with independent random seeds. That is, consider morphisms $\varphi_1(c), \varphi_2(c), \ldots, \varphi_s(c)$ and define a morphism $\varphi(c) = \varphi_1(c) \varphi_2(c) \ldots \varphi_s(c)$ with property: 

\[\expect{\norm{\varphi(c)-\varphi(c')}_H} = s \cdot \expect{\norm{\varphi_i(c)-\varphi_i(c')}_H} = s \cdot  \Theta(\varepsilon^{-2}  |c-c'|^p \frac{1}{p(1-p)}).\]

Assume w.l.o.g. that $(1+\varepsilon)^{\ell-1} < |c-c'| \le (1+\varepsilon)^{\ell}$. We proceed to bound
\begin{align*}
\var{\norm{\varphi(c)-\varphi(c')}_H} &\le s \cdot \sum_{i=\ell+1}^{q} (d_i)^2 \cdot \Pr[c_i \not= c'_i] \le \\
&\le s \cdot \sum_{i=\ell+1}^{q} \varepsilon^{-2} ((1+\varepsilon)^{2p})^i \frac{|c-c'| }{(1+\varepsilon)^i} \le \\
&\le s \cdot \varepsilon^{-2} |c-c'| \sum_{i=\ell+1}^{\infty} ((1+\varepsilon)^{2p-1})^{i} \le \\
&\le s \cdot \varepsilon^{-2} |c-c'|^{2p} \frac{(1+\varepsilon)^{2p-1}}{1 - (1+\varepsilon)^{2p-1}} \le \\
&=  s \cdot \bigo(|c-c'|^{2p}  \varepsilon^{-3} \frac{1}{1-2p}).
\end{align*}
We set $s = \Theta(\frac{|c-c'|^{2p}  \varepsilon^{-3} (p(p-1))^2}{\varepsilon^{2}(|c-c'|^p \varepsilon^{-2})^2(1-2p)} ) = \bigo(\varepsilon^{-1} \frac{1}{1-2p})$ for the claim to hold via Chebyshev's inequality.
The error probability coming from Chebyshev's inequality can be made arbitrarily small constant by fixing the constant factor in $s$ to be large enough. We finally set $t = sq$.

\textbf{Case $p=1/2$.} Note that for $p,p'$ such that $|p-p'| \le \log_{\sigma} (1+\eps)$ we have $|x|^p \eqeps |x|^{p'}$ for all $-\sigma \le x \le \sigma$. We can therefore reduce this case to $p = 1/2 - \log_{\sigma} (1+\eps)$. However, we have to take into account that the asymptotic growth of $t$ hides $1/(1-2p)$ dependency on $p$ for $0 < p < 1/2$, hence $t =  \widetilde\bigo(\varepsilon^{-3})$ for $p = 1/2$. 

\textbf{Case $1/2 < p < 1$.} The proof follows the steps of the case $0 < p < 1/2$. We first bound the variance:
\begin{align*}
\var{\norm{\varphi(c)-\varphi(c')}_H} &\le s \cdot \sum_{i=\ell+1}^{q} (d_i)^2 \cdot \Pr[c_i \not= c'_i] =\\
&=s\cdot \varepsilon^{-2} |c-c'| \sum_{i=\ell+1}^{q} ((1+\varepsilon)^{2p-1})^{i} =\\
&= s\cdot\bigo(\varepsilon^{-3} |c-c'| ((1+\varepsilon)^q)^{2p-1}) = \\
&= s\cdot \bigo(\varepsilon^{-3} |c-c'| \sigma^{\frac{2p-1}{1-p}} / \eps^{3 \cdot \frac{2p-1}{1-p}} ).
\end{align*}
We set $s = \Theta\left( \frac{\varepsilon^{-3} |c-c'| \sigma^{\frac{2p-1}{1-p}} / \eps^{3 \cdot \frac{2p-1}{1-p}}  }{\varepsilon^{-2} |c-c'|^{2p}}\right) = \bigo( \sigma^{\frac{2p-1}{1-p}} / \eps^{1+3 \cdot \frac{2p-1}{1-p}})$, so that 
by Chebyshev's inequality, the probability of obtaining $\norm{U-V}_p^p \eqeps \alpha \cdot \norm{\varphi(U)-\varphi(V)}_H  $ is an arbitrarily small constant (by setting $s$ to be large enough).

\textbf{Randomness.}
The only source of randomness in the description are the values $r_i \in [0,1]$ picked u.a.r. and independently. We note that the values $r_i$ can be picked instead as a finite precision floating-point numbers. Since all the values we are working with are bounded by $\textrm{poly}(n\sigma\varepsilon^{-1})$, it is enough to set precision accordingly.
We also  observe that our concentration argument involves only Chebyshev's inequality and thus only the variance and the expected value, so it suffices to require that $r_i$ are $4$-wise independent.
\end{proof}

We now describe how to use the morphism $\varphi$ to approximate the $L_p$ distances in a small space. To design an efficient algorithm, we take advantage of the fact that $\varphi(U)$ has a compressed representation of size comparable with the length of $U$ (at least when $p \le 1/2$). 

\begin{definition}[$L_p$ scaling]
Consider a word $S = s_1^{e_1} s_2^{e_2} \ldots s_{m}^{e_{m}}$ of length $m' = \sum_i e_i$.
Let $h : [m] \to 2^q$ be a function drawn at random from a $4$-wise independent family, where $q = \ceil{ 3 \log m'}$. For $r = 0, \ldots, q$, we define the $r$-th level subsample of $S$, 
\[\rSub_r(S) = (s_1)^{\left\lfloor \frac{e_1 + (h(1) \bmod 2^r)}{2^r}\right\rfloor} (s_2)^{\left\lfloor \frac{e_2 + (h(2) \bmod 2^r)}{2^r}\right\rfloor} \ldots (s_m)^{\left\lfloor \frac{e_m + (h(m) \bmod 2^r)}{2^r}\right\rfloor}\] 
In particular, $\rSub_0(U) = U$.
\end{definition}

Consider two words $S, Q$ of form $S = s_1^{e_1} \ldots s_m^{e_m}$ and $Q = q_1^{e_1} \ldots q_m^{e_n}$.
Fix an integer $k = \Theta(1/\eps^2)$ large enough and consider $\rSub_r(S), \rSub_r(Q)$ for all $r = 0, 1, \ldots, \ceil{3 \log m'}$, where $m' = \sum_i e_i$.
\begin{algo}\label{alg:lppref}\ 
\begin{enumerate}
\item Denote $X_r = \norm{\rSub_r(S) - \rSub_r(Q)}_H$ and let $f = \min\{ i : X_i \le k\}$.
\item Output $Z_f = 2^f \cdot X_f$ as an estimate of $\norm{S-Q}_H$.
\end{enumerate}
\end{algo}

\begin{lemma}
For $Z_f$ as in Algorihtm~\ref{alg:lppref} there is $Z_f \eqeps \norm{S-Q}_H$ with probability $\ge 3/4$.
\end{lemma}
\begin{proof}
Consider a fixed subsampling level $r$. For simplicity, let $\rSub_r(S) = s_1^{e_1'} s_2^{e_2'} \ldots s_m^{e_m'}$ and $\rSub_r(Q) = q_1^{e_1'} q_2^{e_2'} \ldots q_m^{e_m'}$. Define a random variable $x_i$ to be the contribution of of $s_i^{e'_i}, q_i^{e'_i}$ to the Hamming distance $X_r$, i.e. 
\[x_i = \norm{s_i^{e'_i} - q_i^{e'_i}}_H = e'_i \cdot \norm{s_i - q_i}_H.\]
Since $e'_i \in \{ \lceil e_i/2^r \rceil, \lfloor  e_i/2^r \rfloor \}$ and $\expect{e'_i} = e_i/2^r$, we have $\expect{x_i} = e_i \cdot \norm{s_i - q_i}_H$ and 
\[\var{x_i} = \var{x_i - \lfloor  e_i/2^r \rfloor} \le \expect{x_i - \lfloor  e_i/2^r \rfloor } \le \expect{x_i}.\]
Summing over all values of $i$, we reach $\expect{X_r} = \norm{S-Q}_H$ and $\var{X_r} \le \expect{X_r}$. These bounds are identical to that of Lemma~\ref{lem:hamming_subsampling} and we can proceed in a similar fashion to obtain the claim.
\end{proof}

We are now ready to define $L_p$ prefix encodings. Consider a $b$-length block $B$ of the text and define $B_r^j = \rSub_r(\varphi(B))[(b-j)t+1,bt]$ ($t$ is defined as in Theorem~\ref{th:wordmorphism_p12}). Also, define $P_r^j$ to be the $(tj)$-length prefix of $\varphi(P)$ subsampled in a synchronized way with $B_r^j$.

\begin{definition}
Consider a $b$-length block $B$ of the text $T$. For each $r = 0, \ldots, \ceil{ 3 \log n'}$, where $n' = |\varphi(B)|$, let $j^*(r)$ be the maximal integer such that the Hamming distance between $B_r^{j^*(r)}$ and $P_r^{j^*(r)}$ is at most $k = \Theta(\varepsilon^{-2})$. We define the $L_p$ prefix encoding of $B$ to be a tuple of pairs $j^*(r), \MI(B_r^{j^\ast (r)}, P_r^{j^\ast (r)})$. 
\end{definition}

The $L_p$ prefix encoding of $B$ uses $\bigo(k \log n') = \bigo( \eps^{-2} \log (n \sigma \eps^{-1}))$ space. 

\begin{lemma}\label{lm:rleprefsketchalg}
Assume constant-time random access to $P[1,b]$. Given a $b$-length block~$B$ of the text $T$, its $L_p$ prefix encoding can be computed in $\bigo(b^2 \cdot t \log {n \sigma \eps^{-1}})$ time and $\bigo(b+\eps^{-2}\log {n \sigma \eps^{-1}})$ space.
\end{lemma}
\begin{proof}
For each $r = 0, \ldots, \ceil{ 3 \log n'}$ and $j = 1,\ldots, b$, we compute the Hamming distance between $B_r^j$ and $P_r^j$ in $\Oh(bt)$ time using the compressed representation to find $j^\ast(r)$ and the corresponding mismatch information. The claim follows. 
\end{proof}

\begin{lemma}\label{lm:rlepref_decoding}
Given the $L_p$ prefix encoding of a $b$-length block $B$ of the text $T$, there is an algorithm that computes, for all $j = 1, \ldots, b$, a $(1\pm\eps)$-approximation of the $L_p$ distance between the $j$-length suffix of $B$ and the $j$-length prefix of $P$ in $\widetilde\Oh(b^{2} \cdot t \log {n \sigma \eps^{-1}})$ time and $\Oh(b+\eps^{-2}\log {n \sigma \eps^{-1}})$ space.
\end{lemma}
\begin{proof}
Denote $X_r = \norm{P_r^j - B_r^j}_H$. We compute the smallest $f$ such that $X_f \le k$ in the following way. For each $r$, we use $\MI(B_r^{j^\ast (r)}, P_r^{j^\ast (r)})$ to restore $B_r^{j^\ast (r)}$. If $j > j^\ast (r)$, the Hamming distance between $P_r^j$ and $B_r^j$ is at least $k$. Otherwise, we compare $P_r^j$ and $B_r^j$ to compute the Hamming distance in $\Oh(b t )$ time. The claim follows.
\end{proof}

\subsection{Generic (\texorpdfstring{$L_p$}{Lp}) distance for \texorpdfstring{$1 < p \le 2$}{1 < p <= 2}.}\label{sec:prefix_lp_l2}
For $1 < p \le 2$, we use a scheme similar to the one developped in~\cite{HDstream} for the Hamming distance, but adapt it to generic $L_p$ distances. Particularly, we plug in a standard tool used in this situation, the $p$-stable distribution. We additionally have to adapt the scheme a bit, taking into account that $L_p$ norm is sub-additive under concatenation when $p>1$.

\begin{definition}[$p$-stable distribution~\cite{zolotarev1986one}]
\label{def:pstable}
For a parameter $p>0$, we say that a distribution~$\mathcal{D}$ is \emph{$p$-stable} if for all $a,b \in \mathbb{R}$ and random variables $X, Y$ drawn independently from $\mathcal{D}$, the variable $a X + b Y$ is distributed as $\left(|a|^p + |b|^p \right)^{1/p} Z$, where $Z$ is a random variable with distribution $\mathcal{D}$. 
\end{definition}

Consider a word $X = x_1 x_2 \ldots x_n$, and let $\alpha_1, \alpha_2, \ldots,\alpha_n$ be independent random variables drawn from a $p$-stable distribution $\mathcal{D}$ with expected value $\mu_{\mathcal{D}}$. By Definition~\ref{def:pstable}, we have $\expect{\sum_i \alpha_i x_i} / \mu_{\mathcal{D}} = \norm{X}_p$. The $p$-stable distributions exist for all $0< p \le 2$, and a random variable $X$ from a $p$-stable distribution can be generated using the formula $X = \frac{\sin(p \Theta)}{\cos^{1/p} (\Theta)} \left( \frac{\cos(\Theta(1-p))}{\ln (1/r)}\right)^{(1-p)/p}$~\cite{CMS76,zolotarev1986one}, where $\Theta$ is uniform on $[-\pi/2, \pi/2]$ and~$r$ is uniform on $[0,1]$.

However, to be able to design an efficient sketching scheme that allows to approximate the $L_p$ norm with high probability, there are three technicalities  to be overcome: First, one must show that $\sum_i \alpha_i x_i$ concentrates well, second, the formula above assumes infinite precision of computation, and finally, one cannot use fully independent random variables $\alpha_i$ as above as this would require much space. To overcome these issues, Indyk~\cite{DBLP:journals/jacm/Indyk06} combined $p$-stable distributions and pseudorandom generators for bounded space computation~\cite{DBLP:conf/stoc/Nisan90}. We restate the final result of Indyk below, in the form that will be convenient for us later.

\begin{theorem}[cf. Theorem 2, Theorem 4~\cite{DBLP:journals/jacm/Indyk06}]
\label{th:psketch}
For any $0 < p \le 2$, there is a non-uniform streaming algorithm that maintains a sketch $\pSketch(S)$ of a word $S$ of length $n$ over an alphabet of size~$\sigma$ such that:
\begin{enumerate}
\item  when a new character of $S$ arrives, the sketch can be updated in $\Oh(\eps^{-2} \log (n/\eps))$ time;
\item the algorithm and the sketch use $\Oh(\eps^{-2} \log(\sigma n/\eps)\log(n/\eps))$ bits of space.
\end{enumerate}
Given the sketches $\pSketch(X), \pSketch(Y)$ of two words $X, Y$ of length $n$, one can estimate $\norm{X-Y}_p$ up to a factor $1\pm\eps$ with probability at least $9/10$ in time $\widetilde{\Oh}(1/\eps^2)$. 
\end{theorem}

We now proceed to building the $L_p$ prefix encoding by using $\pSketch$ and the landmarking technique.

\begin{definition}[$L_p$ prefix encoding]
Let $1 < p \le 2$.
Consider a word $S$ of length $b$ on the alphabet of size~$\sigma$. Define $q_0 = b$. For $k = 0, \ldots, \ceil{\log b \sigma^p}$, let $q_k \le q_{k-1}$ be the leftmost position such that the $p$'th moment of the difference between $S[q_k, b]$ and $P[1,b-q_k+1]$, i.e. $\norm{S[q_k, b]-P[1,b-q_k+1]}_p^p$, is at most~$2^k$. 

Further, divide $S[q_k, b]$ into $\Theta(1/\eps^{ p})$ blocks such that each block is either a single character, or the $p$'th moment of the difference between each block and the corresponding subword of $P[1,b-p_k+1]$ is at most $\eps^{p} \cdot 2^{k}$. Let $q_k = q_k^0 \le q_k^1 \le \ldots q_k^{\ell_k} = b$ be the block borders. We choose $q_k^1, q_k^2, \ldots, q_k^{\ell_k}$ from left to right, and each position $q_k^i$ is chosen to be the rightmost possible. 

The \emph{$L_p$ prefix encoding} of $S$ is defined to contain sorted lists of the positions $q_k$ and $q_k^i$, characters $S[q_k^i]$, and sketches for $(1\pm C_p \cdot\eps / p)$-approximating the $p$'th norm of $S[q_k^j, b]$, for all $k, j$ and $C_p$ as in Observation~\ref{obs:norm_moment}, see also Theorem~\ref{th:psketch}.
\end{definition}

The encoding takes $\widetilde{\Oh}(\eps^{-2- p} \log \sigma \log (\sigma n/\eps) \log(n/\eps))$ bits of space. We now show that given the $L_p$ prefix encoding of a block $B$ of the text of length $b$, one can compute a $(1\pm\eps)$-approximation of the $L_p$ distance between any prefix $P[1,b-j+1]$ of the pattern $P$ and the corresponding suffix $B[j,b]$ of $B$. 

\begin{figure}
\begin{center}
\begin{tikzpicture}[scale=0.3]
	\draw[thick] (0,5) rectangle (42,6);
	\draw[draw=none] (19,5) rectangle (25,6) node[pos=0.5] {\small{$B_1$}};
	\draw[draw=none] (25,5) rectangle (42,6) node[pos=0.5] {\small{$B_2$}};		
	\draw[thick,pattern=north west lines, pattern color=red] (19,1) rectangle (25,2) node[pos=0.5] {\small{$P_1$}};		
	\draw[thick,pattern=north east lines, pattern color=gray] (25,1) rectangle (42,2) node[pos=0.5] {\small{$P_2$}};		
	
	\node[below] at (0.5,4.8) {\small $1$};
	\node[below] at (13,4.6) {\small $q_{k+1}$};
	\node[below] at (17,5.2) {\small $q_k^{i}$};	
	\node[below] at (19.5,4.9) {\small $j$};
	\node[below] at (26.5,5.2) {\small $q_k^{i+1}$};
	\node[below] at (30,4.6) {\small $q_{k}$};
	\node[below] at (41.5,4.8) {\small $b$};
	\node[below] at (41.5,0.8) {\small $b-j+1$};
	\node[below] at (19.5,0.8) {\small $1$};
	
	\draw[dashed] (19,6)--(19,1);
	\draw[dashed] (25,6)--(25,1);
\end{tikzpicture}
\end{center}
\caption{Using the prefix encoding of $B$ to compute the $L_p$ distance between a suffix of $B$ and a prefix of the pattern. To compute the distance between $B_1$ and $P_1$, we replace $B_1$ with a subword of the pattern, and between $B_2$ and $P_2$ we use the sketches.} 
\label{fig:psketch}
\end{figure}
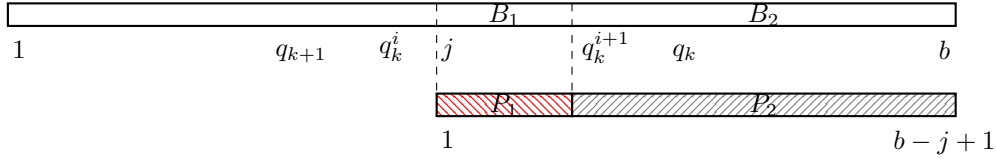
\begin{lemma}
\label{lem:holder}
Let $1 < p \le 2$. For any two vectors $X, Y$ of equal length, $\Big| \norm{X+Y}_p^p -  \norm{X}_p^p \Big| = \bigo(\norm{Y}_p^p +  \norm{Y}_p \cdot \norm{X}_p^{p-1})$. 
\end{lemma}
\begin{proof}
Consider $x, y \in \mathbb{R}$. If $|x| \ge |y|$, then by Taylor expansion, $|x+y|^p = |x|^p (1+y/|x|)^p = |x|^p (1+\bigo(|y/x|)) = |x|^p \pm \bigo(|y| |x|^{p-1})$. If $|x| < |y|$, then $|x+y|^p = \bigo(|y|^p)$. Thus for any real values, we have 
\[|x+y|^p = |x|^p + \bigo(|y|^p + |y| \cdot |x|^{p-1}).\]
Denote $X = [x_1, x_2, \ldots, x_n]^T$ and $Y = [y_1, y_2, \ldots, y_n]^T$. There is
\[\norm{X+Y}_p^p = \sum_i |x_i + y_i|^p = \sum_i |x_i|^p \pm \bigo\left( \sum_i |y_i|^p + \sum_i |y_i| |x_i|^{p-1}\right).\]
Pick $q = p/(p-1)$ so that $1/p + 1/q = 1$. By Hölder's inequality:
\[\sum_i |y_i| |x_i|^{p-1} \le \left(\sum_i |y_i|^p \right)^{1/p} \left( \sum_i |x_i|^{(p-1)q}\right)^{1/q} = \norm{Y}_p \norm{X}_p^{p-1}.\]
\end{proof}

\begin{lemma}\label{lm:lp_prefix_norm}
Let $1 < p \le 2$. Given the $L_p$ prefix encoding of a block $B$ of the text $T$ of length $b$, one can find $(1\pm\eps)$-approximation of the $p$'th moment of the difference between any prefix $P[1,b-j+1]$ of the pattern~$P$ and the corresponding suffix $B[j,b]$ of $B$ in time $\widetilde{\Oh}(\eps^{-2} + \log \sigma)$.
\end{lemma}
\begin{proof}
Let $q_k$ be the position that is closest to $i$ from the left, and $q_k^i \le j < q_k^{i+1}$ (see Fig.~\ref{fig:psketch}). We can find $q_k$, $q_k^i, q_k^{i+1}$ in time $\Oh(\log (b\sigma^p) + 1/\eps^{p})$ by iterating over the sorted lists.

The position $q_k^{i+1}$ divides $P[1,b-j+1]$ into two parts, $P_1$ and $P_2$. Denote $B_1$ and $B_2$ the respective subwords of $B$ they are aligned with  (see Fig.~\ref{fig:psketch}). Let $m_1 = F_p(P_1 - B_1)$ and $m_2 = F_p(P_2 - B_2)$. Then $m = F_p(P[1,b-j+1] - B[j,b])$, being the value we need to approximate, is equal to $m_1+m_2$.  

We can find $m'_2 \eqeps m_2$ using the sketches for $B_2 = B[q_k^{i+1},b]$ and $P_2$ in time $\widetilde{\Oh}(1/\eps^2)$. Furthermore, if $q_k^i = q_k^{i+1}-1$, then we can compute $m_1$ exactly as we store $B[q_k^i]$. Otherwise, we consider the subword $\tilde{P} = P[j-q_k+1,q_k^{i+1}-q_k+1]$ of the pattern $P$. Denote $m'_1 = F_p(P_1-\tilde{P})$ and use it as our estimation of~$m_1$.

Since $1 \le p \le 2$, by definition, $F_p(B_1 - \tilde{P}) \le \varepsilon^p \cdot 2^{k-1}$, and $F_p(P_1-B_1) \le 2^k$.
By Lemma~\ref{lem:holder} with $X = P_1 - B_1$ and $Y = B_1 - \tilde{P}$, 
\[|m'_1 - m_1| = \bigo(\norm{B_1-\tilde{P}}_p^p + \norm{B_1-\tilde{P}}_p \norm{P_1-B_1}_p^{p-1}) = \bigo(\varepsilon^p 2^k + \varepsilon (2^k)^{\frac1p} (2^k)^{\frac{p-1}{p}}) = \bigo(\varepsilon 2^k)\]
and finally $|(m_1+m_2) - (m'_1 + m'_2)| \le\bigo(\varepsilon m) + \varepsilon m_2 = \bigo(\varepsilon m)$.
\end{proof}

\begin{lemma}\label{lm:prefix_lp}
Let $1 < p \le 2$. The $L_p$ prefix encoding of a $b$-length block $B$ of the text can be computed in time $\tilde{\Oh}(b^2 + \eps^{-2} b \log \sigma)$ and space $\tilde{\Oh}(b +\eps^{-2- p} \log^2 \sigma)$. 
\end{lemma}
\begin{proof}
For $j = 1, \ldots, b$, we naively compute the $L_p$ distance between the suffix of $B$ and the prefix of $P$ in $\Oh(b)$ time. We then find the positions $q_k$. For each $k  = 0, \ldots, \ceil{\log b \sigma^p}$, we can find the positions $q_k^i$ in $\Oh(b)$ time and compute the sketches in $\tilde{\Oh}(\eps^{-2} b)$ time by Theorem~\ref{th:psketch}.
\end{proof}

\section{Suffix sketches}
In this section, we give the definitions and explain how we maintain the suffix sketches for each of the distances.

\subsection{Hamming  distance}\label{sec:suffix_l0}
We first recall Euclidean suffix sketches as presented in~\cite{HDstream}. In fact, we will not use them for the Euclidean distance as for it we can use the generic solution of Section~\ref{sec:suffix_lp}, but they will serve as a foundation of Hamming suffix sketches. 

All sketches presented in this section are correct with constant probability, which can be amplified to  $1-\delta$ for arbitrarily small $\delta$ by a standard method of repeating sketching independently $\Theta(\log \delta^{-1})$ times and taking the median of the estimates.

\begin{lemma}[Euclidean sketches~\cite{DBLP:journals/jcss/Achlioptas03}]\label{lm:esketch}
Let $M$ be a random matrix of size $d \times n$ filled with 4-wise independent random $\pm 1$ variables, for $d = \Theta(\eps^{-2})$ chosen big enough. For a vector $X \in \mathbb{R}^n$ there is $\frac{1}{\sqrt{d}} \norm{M X}_2 \eqeps \norm{X}_2$ with constant probability $9/10$, taken over all possible choices of $M$. We say that a vector $M X$ of dimension~$d$ is a Euclidean sketch of $X$.
\end{lemma}

\begin{definition}[Euclidean suffix sketches~\cite{HDstream}]\label{def:esufsketch}
Consider a word $X$ of length $n$. We define its Euclidean suffix sketch as follows. 

Let $b$ be the \emph{block length}. Let  $\mathcal{R}$ be a random matrix of size $d \times b$ filled with 4-wise independent random $\pm 1$ variables and let $\alpha_1, \ldots, \alpha_{\lceil n/b\rceil}$ be 4-wise independent random coefficients with values $\pm1$ as well. We define a matrix $M$ of size $d \times n$ such that $M_{i,j b + k} = \alpha_j \cdot \mathcal{R}_{i,k}$.

Let $X'$ be a word of length $\lceil n/b\rceil \cdot b$ obtained from $X$ by appending an appropriate number of zeroes. The \emph{Euclidean suffix sketch} of $X$ is defined as $\eSketch(X) = M X'$, where $X'$ is considered as a vector.
\end{definition}

Observe that the matrix $M$ does not need to be accessed explicitly. Indeed, from $M X' = \sum_i \alpha_i \cdot \mathcal{R} \cdot \begin{bmatrix} X'[bi],\ \ldots,\ X'[bi+b-1]\end{bmatrix}^T$ it follows that the Euclidean suffix sketch can be computed by first sketching each block of $X'$ using the matrix $\mathcal{R}$, and then taking a linear combination of the sketches of the blocks (using the random $\pm 1$ coefficients $\alpha_i$).

\begin{lemma}[\cite{HDstream}]\label{lm:esufsketch}
Selecting $d = \Theta(\eps^{-2})$ gives $\frac{1}{\sqrt{d}} \norm{\eSketch(X)}_2 \eqeps \norm{X}_2$ with probability at least $9/10$ (taken over all possible choices of $\mathcal{R}, \alpha_i$).
\end{lemma}

By linearity of sketches, we obtain $\norm{X - Y}_2 \eqeps \frac{1}{\sqrt{d}} \norm{\eSketch(X)-\eSketch(Y)}_2$ with probability at least $9/10$ as well.\

We now define Hamming suffix sketches. First note that for binary words $X, Y$ there is $\ham(X,Y) = \norm{X-Y}_2$, and therefore in the case of the binary alphabet we can use the Euclidean suffix sketches. We will now show how to reduce the case of arbitrary polynomial-size alphabets to the case of the binary alphabet. 

To this end,~\cite{HDstream} used a random mapping of Karloff~\cite{Karloff93} as a black-box reduction, which led to sketches of size $\sim \eps^{-4}$ . 
We now show a more careful reduction to avoid this overhead and to achieve dependency $\eps^{-2}$ in total. Consider a word morphism defined on alphabet as $\mu : \Sigma \to \{0,1\}^{\sigma}$, $\mu(a) = 0^{a}10^{\sigma-a-1}$ (and acting on words by concatenating the images of each character of the input word). Note that $\norm{\mu(X)-\mu(Y)}_2^2  = 2 \cdot\norm{X-Y}_H$, thus using the Euclidean suffix sketches on top of $\mu(X)$ and $\mu(Y)$ allows computation of the respective Hamming distance. Formally,

\begin{definition}[Hamming suffix sketches~\cite{HDstream}]\label{def:hsufsketch}
Consider a word $X$ of length $n$ on the alphabet of size~$\sigma$. We define its Hamming suffix sketch as follows. 

Let $b$ be the \emph{block length}, $\mathcal{R}$ be a random matrix of size $d \times \sigma b$ filled with 4-wise independent random $\pm 1$ variables, and $\alpha_1, \ldots, \alpha_{\lceil n/b\rceil}$ be 4-wise independent random coefficients with values $\pm1$ as well. We define a matrix $M$ of size $d \times \sigma n$ such that $M_{i,\sigma j b + k} = \alpha_j \cdot \mathcal{R}_{i,k}$.

Let $X'$ be a word of length $\lceil n/b\rceil \cdot b$ obtained from $X$ by appending an appropriate number of zeroes. The \emph{Hamming suffix sketch} of $X$ is defined as $\hSketch(X) = M \mu(X')$, where $\mu(X')$ is considered as a vector.
\end{definition}

\begin{lemma}
Selecting $d = \Theta(\eps^{-2})$ gives $\frac{1}{2d} \norm{\hSketch(X)}_2^2 \eqeps\norm{X}_H$ with probability at least~$9/10$ (taken over all possible choices of $\mathcal{R}, \alpha_i$).
\end{lemma}
\begin{proof}
Follows immediately as a corollary of Lemma~\ref{lm:esufsketch} and the properties of the embedding~$\mu$. In more detail, the following holds with probability at least $9/10$: 
\begin{align*}
&\frac{1}{2d} \cdot \norm{\hSketch(X)}_2^2 =\\
&= \frac1{2d} \norm{M\mu(X')}_2^2 = \frac1{2d} \norm{M\mu(X)}_2^2 = \frac{1}{2d} \norm{\eSketch(\mu(X))}_2^2 \eqeps \frac{1}{2} \norm{\mu(X)}_2^2  = \norm{X}_H.
 \end{align*}
\end{proof}

As $\mu(X), \mu(Y)$ are sparse, there is an efficient streaming algorithm for maintaining the Hamming suffix sketches of a text: 

\begin{lemma}
\label{lm:hsufsketchalg}
Given a text $T$, there is a streaming algorithm that for every position $i$ outputs the Hamming suffix sketch of a word $T[b\cdot k + 1, i]$, where $k$ is the largest integer such that $i - b \cdot k \le n$. The algorithm takes $\bigo(d n/b + \log d \sigma n)$ space and $\bigo(d(1+n/b^2))$ time per character. 
\end{lemma}
\begin{proof}
We fix the matrix $\mathcal{R}$ and the random coefficients $\alpha_1, \ldots, \alpha_{\lceil n/b\rceil}$ from Definition~\ref{def:hsufsketch}. We do not store $\mathcal{R}$ and $\alpha_i$ explicitly, but generate them using two hash functions drawn at random from a $4$-wise independent family. For example, to generate $\mathcal{R}$ we can consider a family of polynomials $2((ax^3 +bx^2 +cx+d \bmod p) \bmod 2) - 1$, with parameters $a, b, c, d$ chosen u.a.r. from the prime field~$\mathbb{F}_p$ for $p \ge db$, and $\alpha_i$ can be generated in a similar fashion. This way, we need to store only $\Oh(\log(d\sigma b) + \log(n/b)) = \Oh(\log d \sigma n)$ random bits that define the coefficients of two polynomials to generate $\mathcal{R}$ and $\alpha_i$.

We process the text $T$ by blocks $B_1, B_2, \ldots$ of length $b$. For each block $B_k$ we compute its sketch using the matrix $\mathcal{R}$. That is, at the beginning of each block we initialize its sketch as a zero vector of length $d$. When a new character $T[i]$ of a block $B_k$ arrives, we compute and add $[M[1, i\cdot b \sigma + T[i]], M[2, i\cdot b \sigma + T[i]], \ldots, M[d, i\cdot b  \sigma+ T[i]]]^T$ to the sketch, which takes $\Oh(d)$ time. We store the sketch of $B_k$ until the block $B_{k+\ceil{n/b}}$ and use it to compute the suffix sketches for the positions in this block.

Consider now a block $B_{k+\lceil n/b \rceil}$. We first compute the suffix sketch for the position $b \cdot (k+\ceil{n/b})$, which is the position preceding the block $B_{k+\ceil{n/b}}$. The suffix sketch for it is simply a linear combination of the sketches of the blocks $B_{k+\ceil{n/b}-1}, B_{k+\ceil{n/b}-2}, \ldots, B_k$ with coefficients $\alpha_1, \ldots, \alpha_{\lceil n/b\rceil - 1}$. Since each sketch is a vector of length $d$, we can compute the linear combination in $\Oh(d n/b)$ time. To make this computation time-efficient, we start it $b$ positions before position $b \cdot (k+\ceil{n/b})$ arrives, and de-amortise the computation over these $b$ positions. This way, we use only $\Oh(dn/b^2)$ time per character.

Now, using the suffix sketch for the position $b \cdot (k+\ceil{n/b})$, we can compute the suffix sketches for all positions in the block $B_{k+\ceil{n/b}}$ one-by-one, using only $\Oh(d)$ time per character: When a new character $T[i]$ arrives, we add $[\alpha_{\lceil n/b\rceil} M[1, i\cdot b \sigma + T[i]], \alpha_{\lceil n/b\rceil}  M[2, i\cdot b \sigma + T[i]], \ldots, \alpha_{\lceil n/b\rceil}  M[d, i\cdot b \sigma + T[i]]]^T$ to the suffix sketch to update it. 

Note that at any time we store $\Oh(n/b)$ sketches of the blocks, so the algorithm uses $\Oh(dn / b + \log d \sigma n)$ space in total.
\end{proof}

\subsection{Manhattan (\texorpdfstring{$L_1$}{L1}) distance}\label{sec:suffix_l1}
To show efficient suffix sketches for the Manhattan distance, we consider a word morphism  $\nu : \Sigma \to \{0,1\}^{\sigma}$, $\nu(a) = 1^{a}0^{\sigma-a}$. Note that $\norm{\nu(X)-\nu(Y)}_2^2 = \norm{\nu(X)-\nu(Y)}_H = \norm{X-Y}_1$, thus using the Hamming suffix sketches on top of $\nu(X)$ and $\nu(Y)$ allows computation of the respective Manhattan distance. 

However, if we apply the morphism straightforwardly, we will have to pay an extra $\sigma$ factor per character to compute the Manhattan suffix sketches. To improve the running time, we will use range-summable hash functions. Range-summable hash functions were introduced by Feigenbaum et al.~\cite{FKS+02}, and later their construction was improved by Calderbank et al.~\cite{DBLP:conf/soda/CalderbankGLMS05}.

\begin{definition}[cf.~\cite{DBLP:conf/soda/CalderbankGLMS05}]
A family $\mathcal{H}$ of hash functions $h(x; \xi) : [t] \times \{0,1\}^s \rightarrow \{-1,1\}$ (here~$x$ is the argument and $\xi$ is the seed) is called $k$-independent, range-summable if it satisfies the following properties for any $h \in \mathcal{H}$:
\begin{enumerate}
\item \emph{($k$-independent)} for all distinct $0 \le x_1, \ldots , x_k < t$ and all $b_1, \ldots , b_k \in \{-1, +1\}$, 
\[\Pr_{\xi \in \{0,1\}^s} [h(x_1; \xi) = b_1 \wedge \cdots \wedge h(x_k; \xi) = b_k] = 2^{-k}\]

\item \emph{(range-summable)} there exists a function $g$ such that given a pair of integers $0 \le \alpha, \beta \le \sigma$, and a seed~$\xi$, the value
$g(\alpha, \beta; \xi) = \sum_{\alpha \le x < \beta} h(x; \xi)$
can be computed in time polynomial in $\log t$.\footnote{In~\cite{DBLP:conf/soda/CalderbankGLMS05}, the function $h$ was defined to take values in $\{0,1\}$. We can change the range of values to $\{-1, +1\}$ by taking $h' = 1 - 2h$ while preserving the properties.} 
\end{enumerate}
\end{definition}

\begin{corollary}[cf. Theorem 3.1~\cite{DBLP:conf/soda/CalderbankGLMS05}]
\label{cor:range-summable}
There is a $4$-independent, range-summable family of hash functions $h(x; \xi) : [t] \times \{0,1\}^s \rightarrow \{-1,+1\}$ with a random seed $\xi$ of length $s = \Oh(\log^2 t)$ such that any range-sum $g(\alpha, \beta; \xi)$ can be computed in $\Oh(\log^3 t)$ time. 
\end{corollary}
\begin{observation}\label{obs:sum_of_g}
For a word $X = x_1 x_2 \ldots x_n$, let $Y = \nu(X) = y_1 y_2 \ldots y_{n \sigma}$. Let $h,g$ be as in Corollary~\ref{cor:range-summable} with $t = n\sigma$. Then $\sum_{i=1}^{n} g(i \sigma,i \sigma + x_i;\xi) = \sum_{i=1}^{n \sigma} y_i h(i; \xi)$.
\end{observation}

Thus, we see that range-summable hash functions can be used to efficiently simulate $\nu$.

\begin{definition}[Manhattan suffix sketches]\label{def:msufsketch}
Consider $X$ be a word of length $n$. We define its Manhattan suffix sketch as follows. 

Let $b$ be the \emph{block length}. Let $h,g$ be as in Corollary~\ref{cor:range-summable} with $t = b  d \sigma$. Let $\mathcal{R}$ be a random matrix of size $d \times \sigma b$ filled with 4-wise independent random $\pm 1$ variables, such that $\mathcal{R}_{i,k} = h(i b \sigma + k; \xi)$ and let $\alpha_1, \ldots, \alpha_{\lceil n/b\rceil}$ be 4-wise independent random coefficients with values $\pm1$ as well. We define a matrix $M$ of size $d \times \sigma n$ such that $M_{i,\sigma j b + k} = \alpha_j \cdot \mathcal{R}_{i,k} = \alpha_j \cdot h(dk+i; \xi)$.

Let $X'$ be a word of length $\lceil n/b\rceil \cdot b$ obtained from $X$ by appending an appropriate number of zeroes. The \emph{Manhattan suffix sketch} of $X$ is defined as $\mSketch(X) = M \nu(X')$, where~$\nu(X')$ is considered as a vector.
\end{definition}

\begin{lemma}
Selecting $d = \Theta(1/\eps^2)$ gives $\frac{1}{d}\norm{\mSketch(X) }_2^2 \eqeps \norm{X}_1$ with probability at least~$9/10$ (taken over all possible choices of $\alpha_i$ and $\xi$).
\end{lemma}
\begin{proof}
Follows immediately as a corollary of Lemma~\ref{lm:esufsketch} and the properties of the embedding~$\nu$. In more detail, the following holds with probability at least $9/10$: 
\begin{align*}
&\frac{1}{d} \cdot \norm{\mSketch(X)}_2^2 = \\
&=\frac{1}{d} \norm{M\nu(X')}_2^2 = \frac1{d} \norm{M\nu(X)}_2^2 = \frac{1}{d} \norm{\eSketch(\nu(X))}_2^2 \eqeps \norm{\nu(X)}_2^2  = \norm{X}_1.
\end{align*}
\end{proof}

\begin{lemma}
\label{lm:msufsketchalg}
Given a text $T$, there is a streaming algorithm that for every position $i$ outputs the Manhattan suffix sketch of a word $T[b\cdot k + 1, i]$, where $k$ is the smallest integer such that $i - b \cdot k \le n$. The algorithm takes $\Oh(d \cdot (n/b) + \log^2 \sigma)$ space, and $\Oh(d (1+n/b^2) \cdot \log^3 (b d \sigma))$ time per character.
\end{lemma}
\begin{proof}
The proof mirrors the proof of Lemma~\ref{lm:hsufsketchalg}, and we describe the key elements.
We fix the random coefficients $\alpha_1, \ldots, \alpha_{\lceil n/b\rceil}$ and the hash function $h$ from Definition~\ref{def:msufsketch}. As previously, we do not store the coefficients $\alpha_i$ explicitly, but generate them using a hash function drawn at random from a $4$-wise independent family. The matrix $\mathcal{R}$ is already defined by $h$, with the following parameters: it requires $\bigo(\log^2 (b d \sigma))$ bits of seed, and range-sum queries are answered in time $\bigo(\log^3 (b d \sigma))$.

In the sketching of blocks, we proceed in the same manner, except that when a new character $T[i]$ of a block $B_k$ arrives, we compute and add $\sum_{0 \le j < T[i]} [ M[1, i\cdot b \sigma + j], \ldots, M[d, i\cdot b \sigma + j]]^T = \alpha_i  \cdot [ g(b \sigma, b \sigma + T[i]; \xi), g(2 b \sigma, 2 b \sigma+T[i]; \xi), \ldots, g(d \cdot b \sigma, d \cdot b \sigma + T[i]; \xi)]^T$ to the sketch, which takes $\Oh(d \cdot \log^3 (b d \sigma))$ time ($\log^3 (b d \sigma)$ times slower as the corresponding step in Lemma~\ref{lm:hsufsketchalg}).

Consider now a block $B_{k+\lceil n/b \rceil}$. When a new character $T[i]$ arrives, we update the suffix sketch by adding  $\alpha_{\lceil n/b\rceil} \cdot [ g(b \sigma, b \sigma + T[i]; \xi), g(2 b \sigma, 2 b \sigma+T[i]; \xi), \ldots, g(d \cdot b \sigma, d \cdot b \sigma + T[i]; \xi)]^T$ to it. 

All of the operations are $\bigo(\log^3(b d \sigma))$ time slower than  the corresponding steps in Lemma~\ref{lm:hsufsketchalg}, and the memory complexity is increased by the seed size $\bigo(\log^2 (b d \sigma))$ term ($\log^2 b$ and $\log^2 d$ terms get absorbed).
\end{proof}

\subsection{Generic (\texorpdfstring{$L_p$}{Lp}) distance for \texorpdfstring{$0 < p \le 2$}{0 < p <= 2}.}\label{sec:suffix_lp}
For generic $L_p$ distances, we use the approach of \cite{DBLP:journals/jacm/Indyk06} based on $p$-stable distributions.

\begin{corollary}
\label{cor:psketch_suffix}
Given a text $T$, there is a streaming algorithm that for every position $i$ outputs the $L_p$ suffix sketch of a word $T[b\cdot k + 1, i]$, where $k$ is the smallest integer such that $i - b \cdot k \le n$. The algorithm takes $\Oh(\eps^{-2} (n/b) \cdot \log(\sigma n / \eps) \log(n/\eps))$ bits of space and $\Oh(\eps^{-2} (n/b) \log(n))$ time per character.
\end{corollary}
\begin{proof}
We start a new instance of the sketching algorithm of Theorem~\ref{th:psketch} at every block border and continue running it for the next $\ceil{n/b}$ blocks. At each moment, there are $\Oh(n/b)$ active instances of the algorithm. The bounds follow.
\end{proof}

\section{Proof of Theorem~\ref{th:UB}}\label{sec:theorem}
Recall the structure of the algorithms. During the preprocessing, we compute the suffix sketches of suffixes $P[1,n], P[2,n], \ldots, P[b,n]$ of $P$. During the main stage, the text is processed by blocks of length $b$. To compute an approximation of the distance / the $p$'th moment at a particular alignment, we divide the pattern into two parts: a prefix of length at most $b$, and the remaining suffix. We compute an approximation of the distance / the $p$'th moment for both of the parts and sum them up to obtain the final answer. To compute an approximation of the distance / the $p$'th moment between the prefix and the corresponding block of the text, we compute, while reading each block of the text, its prefix encoding, and to compute an approximation of the distance / the $p$'th moment between the suffix and the text, we use the suffix sketches.

\begin{enumerate}
\item \textbf{Hamming ($L_0$) distance.} When we receive a new block of the text, we compute its Hamming prefix encoding using the algorithm of Lemma~\ref{lm:hprefsketchalg} in $\Oh(b)$ space. We de-amortize the computation over the subsequent block and spend $\widetilde\Oh(\eps^{-2})$ time per character. We store the resulting encoding for the next $\Oh(n/b)$ blocks. In total, the encodings require $\widetilde\Oh(\eps^{-2} n/b)$ space. The Hamming suffix sketches of $P[1,n], P[2,n], \ldots, P[b,n]$ occupy $\Oh(\eps^{-2} b)$ space. The algorithm of Lemma~\ref{lm:hsufsketchalg} that computes the suffix sketches takes $\Oh(\eps^{-2} n/b + \log (\eps^{-2} \sigma n))$ space and $\Oh(\eps^{-2} (1+n/b^2))$ time per character. Consider a block starting with position $p$. To compute the Hamming distances between $n$-length subwords that end in this block and the pattern, we apply the following approach. 
First, while reading the block preceding the current one, we decode the Hamming prefix encoding of the block that starts at position $p-n$ using Lemma~\ref{lm:hpref_decoding}. We de-amortize the algorithm to spend $\widetilde\Oh(\eps^{-2})$ time per character. Hence, at the position $i$, we know the $(1\pm\eps)$-approximation between the prefixes of the pattern and the corresponding subwords of the text. At each position, we can compute the Hamming distance between the corresponding suffix of the pattern and the text in $\widetilde\Oh(\eps^{-2})$ time using the Hamming suffix sketch. By taking $b = \sqrt{n}$, we obtain the claim.

\item \textbf{Manhattan ($L_1$) distance.} We proceed analogously to the Hamming distance case. The Manhattan prefix encoding of each block is computed using Lemma~\ref{lm:mprefsketchalg}, in $\widetilde\bigo(b)$ time per character. We store the resulting encoding for the next $\bigo(n/b)$ blocks, giving in total $\widetilde\bigo(\varepsilon^{-2}n/b)$ space. The Manhattan suffix sketches of $P[1,n], P[2,n], \ldots, P[b,n]$ occupy $\Oh(\eps^{-2} b)$ space. Algorithm of Lemma~\ref{lm:msufsketchalg} takes $\widetilde\bigo(\varepsilon^{-2} (b+n/b) + \log^2 \sigma)$ space and $\widetilde\bigo(\varepsilon^{-2}(1+n/b^2))$ time per character. 
For decoding the prefix encoding we use Lemma~\ref{lm:mpref_decoding}, spending $\widetilde\bigo(b)$ time per character. Once again we take $b = \sqrt{n}$, and assume w.l.o.g. $\eps^{-1} \le \sqrt{n}$ (as otherwise we can use a naive algorithm with $\Oh(n)$ space and $\Oh(n)$ time per character).

\item \textbf{Generic ($L_p$) distance for $0< p < 1$.} The $L_p$ prefix encodings of the blocks are computed using Lemma~\ref{lm:rleprefsketchalg}, using $\widetilde\bigo(t \cdot b)$ time per character. We store the resulting encodings for the next $\bigo(n/b)$ blocks, giving in total $\widetilde\bigo(\varepsilon^{-2}n/b)$ space. 
The $L_p$ suffix sketches of $P[1,n], P[2,n], \ldots, P[b,n]$ occupy $\widetilde{\Oh}(\eps^{-2} b \log \sigma)$ space. 
Algorithm of Corollary~\ref{cor:psketch_suffix} computes the $L_p$ suffix sketches for the text in  $\widetilde{\Oh}(\eps^{-2} (n/b) \log \sigma)$ space and $\widetilde{\Oh}(\eps^{-2} n/b)$ time per character. For decoding the prefix encoding we use Lemma~\ref{lm:rlepref_decoding}, spending $\widetilde\bigo(t \cdot b)$ time per character. We take $b = \sqrt{n}$, and substitute $t$ accordingly to Theorem~\ref{th:wordmorphism_p12}.

\item \textbf{Generic ($L_p$) distance for $1 < p < 2$.} Note that for $\eps < 1/n$ we can use a naive algorithm, that is to store $S$ itself in $\Oh(n)$ space. The update takes constant time, and computing the $L_p$ norm takes $\Oh(n)$ time which is better than the guarantees of the theorem for such values of $\eps$. For $\eps \ge 1/n$, the algorithm of Lemma~\ref{lm:prefix_lp} computes the $L_p$ prefix encodings of the blocks in $\widetilde{\Oh}(b + \eps^{-2-p} \log^2 \sigma)$ space and $\widetilde{\Oh}(b+\eps^{-2} \log \sigma)$ time per character. The encodings occupy $\widetilde{\Oh}(\eps^{-2-p} (n/b) \log^2 \sigma)$ space. The $L_p$ suffix sketches of $P[1,n], P[2,n], \ldots, P[b,n]$ occupy $\widetilde{\Oh}(\eps^{-2} b \log \sigma)$ space. Algorithm of Corollary~\ref{cor:psketch_suffix} computes the $L_p$ suffix sketches for the text in $\widetilde{\Oh}(\eps^{-2} (n/b) \log \sigma)$ space and $\widetilde{\Oh}(\eps^{-2} n/b)$ time per character.  Taking $b = \eps^{-p/2} \sqrt{n}$ and assuming w.l.o.g. $\eps^{-1} < \sqrt{n}$, we obtain the claim.
\end{enumerate}

\section{Conclusion}
We pose several open questions. First is whether the time-complexity for $1/2<p<1$ can be improved to not involve any dependency on $\sigma$. For this we need a better technique than bounding variance of the embedding into Hamming distance: in our technique, the tail gets ''too heavy''. Another pressing question is whether for all values of $p>0$ we could improve upon $\sqrt{n}$ time per character.  We also remark that it seems unlikely that an embedding to Hamming space could be used to reduce space complexity for $p>1$: $L_p^p$ does not admit the triangle inequality while the Hamming distance does, and the $L_p$ distance is not additive with respect to concatenation, while the Hamming distance is. 

\bibliography{bib}
\end{document}